\documentclass[10pt,conference]{IEEEtran}
\ifCLASSINFOpdf
\usepackage[pdftex]{graphicx}

\else
\usepackage[dvipdf]{graphicx}
\DeclareGraphicsExtensions{.eps}
\fi
\usepackage{dblfloatfix}

\usepackage{epstopdf}
\usepackage{amsmath} 
\usepackage{amsthm}

\newtheorem{theorem}{Theorem}

\newtheorem{lemma}[theorem]{Lemma}

\usepackage{epstopdf}
\usepackage{color}
\usepackage{graphicx}
\usepackage{caption}
\usepackage{subcaption}
\usepackage{amsmath, amsthm, amssymb} 
\usepackage{upgreek} 
\usepackage{algorithm} 
\usepackage{algpseudocode} 

\usepackage{enumitem}
\floatstyle{plain}
\newfloat{myalgo}{tbhp}{mya}
{\begin{myalgo}[#1]
\centering
\begin{minipage}{#2}
\begin{algorithm}[H]}%
{\end{algorithm}
\end{minipage}
\end{myalgo}}
\usepackage{setspace}

\IEEEaftertitletext{\vspace{-1\baselineskip}}

\begin{document}

%
\title{Multi-hop Relaying with Mixed Half and Full Duplex Relays for Offloading to MEC}
\author{\IEEEauthorblockN{Pavel Mach$^1$, Zdenek Becvar$^1$, Mohammadsaleh Nikooroo$^1$
} 
\IEEEauthorblockA{
$^1$Czech Technical University in Prague, Czech republic, emails: \{machp2, zdenek.becvar\}@fel.cvut.cz}

}


%


\maketitle

\vspace{-1cm}
\begin{abstract}
In this paper, we focus on offloading a computing task from a user equipment (UE) to a multi-access edge computing (MEC) server via multi-hop relaying. We assume a general relaying case where relays are energy-constrained devices, such as other UEs, internet of things (IoT) devices, or unmanned aerial vehicles. To this end, we formulate the problem as a minimization of the sum energy consumed by the energy-constrained devices under the constraint on the maximum requested time of the task processing. Then, we propose a multi-hop relaying combining half and full duplexes at each individual relay involved in the offloading. We proof that the proposed multi-hop relaying is convex, thus it can be optimized by conventional convex optimization methods. We show our proposal outperforms existing multi-hop relaying schemes in terms of probability that tasks are processed within required time by up to 38\% and, at the same time, decreases energy consumption by up to 28\%.
\end{abstract}


\IEEEpeerreviewmaketitle

\section{Introduction}\label{Intro}
The multi-access edge computing (MEC) introduces a concept of offloading computationally demanding tasks from the energy-constrained user equipment (UE) to the MEC server located at the edge of mobile network \cite{Mach17_COMST}. Hence, the task processing delay and/or energy consumption of the UE can be reduced \cite{Fang2022_IoT}. 

Benefits facilitated by MEC can be further augmented by a relaying of the tasks from the UE to the MEC servers via intermediate relay(s). The exploitation of neighboring UEs as relays and, thus, capitalizing on device-to-device (D2D) relaying concept \cite{Mach2022}, helps to minimize the task processing delay \cite{2021Liu_TVT} or increase the number of tasks completed within a required time \cite{2021Peng_TWC}. Moreover, an adoption of unmanned aerial vehicles (UAVs) acting as the relays can improve quality of experience to the UEs \cite{Zheng2020_GBC} or minimize their energy consumption \cite{Diao2021_TVT}. Besides, the use of vehicles as the relays is considered in \cite{Zhang2021} to ensure a reliable offloading from the vehicles in the area without coverage of the MEC servers. \textcolor{black}{Last, also intelligent reflecting surface (IRS) can assist in offloading of tasks to extend coverage \cite{Chen2023}.} 

All above-mentioned works assume only two-hop relaying, i.e., only one relay is used in the offloading process. To fully grab the potential of the relays, multi-hop relaying for the offloading purposes has recently drawn an attention from researches. The multi-hop relaying is addressed from a perspective of balancing the load among MEC servers \cite{2023_TVT}, minimizing the processing delay of the tasks offloaded from the vehicles to the MEC servers \cite{2020_VehCom}-\cite{2023_arxiv} or to other computing vehicles \cite{2023_IoT}\cite{2023_TITS}, or to offload the tasks from one UE to other neighboring computing UEs \cite{2023_TNSM}.  

The primary objective of all existing studies on the offloading with multi-hop relaying is to find a proper route between the offloading UE and the MEC server or other computing UE. All works but \cite{2023_TNSM} assume only less efficient half-duplex (HD) mode adopted at each relay with the task subsequently offloaded over each hop in individual time intervals, thus, increasing communication delay. The paper \cite{2023_TNSM} considers full-duplex (FD) mode, however, the paper fully disregards the problem of self-interference (SI) with which the FD is inevitably plagued \cite{Wang2018_JSAC}. Moreover, none of the existing works optimize multi-hop relaying in terms of radio resource management including \emph{i}) allocation of time slots at each hop, \emph{ii}) allocation of transmission power of the offloading UE as well as relays, and \emph{iii}) allocation of bandwidth at each hop.  

Motivated by the above-mentioned gaps, the objective of this paper is to optimize radio resource management aspects of multi-hop relaying for the task offloading. Since the offloading UE and relays are usually energy-constrained, such as smartphones, UAVs, or internet of things (IoT) devices, we formulate the problem as the minimization of the sum energy consumed by the energy-constrained UEs involved in the multi-hop relaying under the constraint on the maximum processing time of the computing tasks. First, we propose several unique relaying cases combining HD and FD at each relay involved in multi-hop relaying. Note that existing works always assume the same relaying mode at all relays. Second, we adapt the general problem for each multi-hop relaying case and we prove its convexity so that we can solve it in an optimal way. Finally, we demonstrate that the proposal increases the probability of the tasks being processed within required time by up to 38\% and, at the same time, decreases energy consumption by up to 28\% with respect to state-of-the-art works.

\section{System Model}\label{II}
This section first describes the network model. Then, communication and computing models are introduced.

\subsection{Network model}\label{IIa}
We contemplate a scenario with one powerful MEC server located, for example, at the base station (BS). Further, we assume one UE generating highly computationally demanding tasks  offloaded to the MEC servers via multi-hop relaying. We consider the end-to-end relaying link is already established while leaving the joint optimization of relays selection and multiple UEs scenario for future work. In our work, the relays can be any energy-constrained devices, such as smartphones, IoT devices, UAVs, or vehicles. Moreover, we envisage that different types of relays can be exploited at each hop, e.g., the smartphone can be used as the first relay while vehicle or UAV as other relay.  

\subsection{Communication and relaying models}\label{IIb}
We limit our scenario to two relays (labeled as R1 and R2 in Fig. \ref{Fig01}), since this number is sufficient to illustrate benefits of the multi-hop relaying in the offloading while fitting page limit and avoiding cluttering of text and derivations. Still, all the math derivations and proposed relaying principles can be extended to more than two relays as well. \textcolor{black}{We assume single antenna devices as this configuration gives enough insight on the benefits of proposed multi-hop offloading.} 

We consider that each relay can adopt one of the three relaying modes: \emph{i}) \emph{HD}, \emph{ii}) FD with orthogonal bandwidth at each hop (labeled as \emph{FD--Orthogonal}), and \emph{iii}) FD with the same bandwidth utilized at both hops (labeled as \emph{FD--Shared}). All these modes are described in details in the following subsections.
 
\subsubsection{HD}
In case of HD relaying, the task is first sent over one hop and, then, the same task is relayed over the next hop (see Fig. \ref{Fig01}).  
\begin{figure}[t!]
\centering
\includegraphics[scale=0.75]{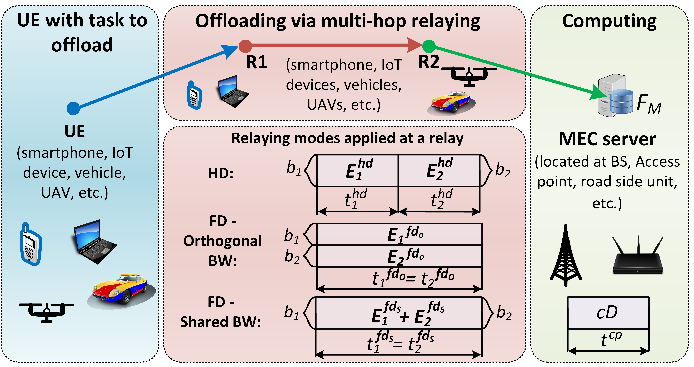}
\caption{System model with one UE having a task to offload via multi-hop relaying and computed at MEC server. } 
\label{Fig01}
\end{figure}
The capacity at the $n$-th hop is:
\begin{equation}
C_{n}^{hd}= b_nlog_2 \left(1+\frac{p_{n} g_{n}}{b_n\left(\sigma+I_b \right)}\right),
\label{C1_HD} 
\end{equation}
where $b_n$, $p_n$, $g_n$ are the allocated bandwidth, transmission power, and channel gain at the $n$-th hop, respectively, $\sigma$ is the noise spectral density, and $I_b$ is the background interference from other UEs in the \textcolor{black}{neighboring cells, as in real-world scenarios, where such interference is usually present.} 

The communication delay in the HD is composed of the delays at individual hops ($t_{1}^{hd}$, $t_{2}^{hd}$), and is expressed as:
\begin{equation}
t^{hd}=t_{1}^{hd}+t_{2}^{hd}=\sum\nolimits_n t_{n}^{hd}=D\sum\nolimits_n \frac{1}{C_{n}^{hd}},
\label{t_HD} 
\end{equation}
where $D$ is the size of the task offloaded by the UE. Similarly as in many works (see, e.g., \cite{2023_TVT}), we neglect the delivery of the computing results back to the UE, as it is insignificant with respect to the whole communication delay.

The sum energy consumed to forward the task in HD is composed of the energies consumed at individual hops ($E^{hd}_n$):
\begin{eqnarray}
E^{hd}=E_{1}^{hd}+E_{2}^{hd}=\sum\nolimits_n E_n^{hd}=\sum\nolimits_n t_{n}^{hd}p_{n}.
\label{E_HD}
\end{eqnarray}

\textcolor{black}{Note that we consider only energy consumption caused by transmission of tasks while a circuit power consumption of the devices is not assumed, as the circuit power consumption is constant and does not change due to offloading.}
 
\subsubsection{FD--Orthogonal}
In this mode, the relay receives and transmits the data simultaneously, but transmission at both hops are orthogonal in frequency domain (see Fig. \ref{Fig01}) to avoid SI. Thus, we assume orthogonal bandwidth $b_1$ and $b_2$ at the first and second hops, respectively (we derive the optimal bandwidth allocation later in the paper). Then, the capacity $C_n^{fd_o}$ at the $n$-th hop is expressed as in (\ref{C1_HD}). 

Since the propagation delay and time for processing of communication at the relay (both jointly denoted as $\epsilon$) are very short (scale of $\mu$s or ms) compared to the offloading time (hundreds of ms or seconds), these can be neglected without breaking a relaying causality and we can assume $t_1^{fd_o}=t_2^{fd_o}+\epsilon \approx t_1^{fd_o}$ for $\epsilon<<$ overall offloading time. In practice, the whole offloading task is transmitted in a series of many smaller transport blocks (in scale of ms in 5G) and each block can be forwarded right after its reception and processing by the relay, hence, fulfilling the relaying causality principle for the whole task. Thus, the communication delay is:
\begin{equation}
t^{fd_o}=t_1^{fd_o}=t_2^{fd_o}=\text{max}(D/C_{1}^{fd_o},D/C_{2}^{fd_o}),
\label{t_FD1} 
\end{equation}

The energy consumption required to relay the task in FD with orthogonal bandwidth is calculated as:
\begin{eqnarray}
E^{fd_o}=E_{1}^{fd_o}+E_{2}^{fd_o}=\sum\nolimits_n E_n^{fd_o}=\sum\nolimits_n t_{n}^{fd_o}p_{n}.
\label{E_FD1}
\end{eqnarray}
 
\subsubsection{FD--Shared}
Similarly as in the previous FD case, the relays can receive and transmit data simultaneously (assuming $\epsilon<<$ offloading time as explained for \emph{FD-Orthogonal bandwidth}). The fundamental difference is, however, that the transmissions at both hops share the same bandwidth. Then, the capacities at each hop are:
\begin{equation}
C_{1}^{fd_s}= b_1 log_2 \left(1+\frac{p_{1} g_{1}}{b_1\left(\sigma+I_b \right)+p_{2}g_{1,1}}\right),
\label{C1_FD2} 
\end{equation}
\begin{equation}
C_{2}^{fd_s}=b_2 log_2\left(1+\frac{p_{2} g_{2}}{b_2\left(\sigma +I_b\right)+p_{1}g_{1,2}}\right),
\label{C2_FD2} 
\end{equation}
where $g_{1,1}$ is the channel gain between the transmitter and the receiver of the relay, thus, $p_{2}g_{1,1}$ in (\ref{C1_FD2}) represents the SI in FD \cite{Wang2018_JSAC}; and $g_{1,2}$ is the channel gain between the transmitter at the first hop and the receiver at the second hop and $p_{1}g_{1,2}$ representing interference in (\ref{C2_FD2}) from the former to the latter.

The communication delay is analogous to FD with orthogonal bandwidth, i.e.,:
\begin{equation}
t^{fd_s}=t_1^{fd_s}=t_2^{fd_s}=\text{max}(D/C_{1}^{fd_s},D/C_{2}^{fd_s}).
\label{t_FD2} 
\end{equation}
The energy consumption in this relaying mode is defined as:
\begin{eqnarray}
E^{fd_s}=E_{1}^{fd_s}+E_{2}^{fd_s}=\sum\nolimits_n E_n^{fd_s}=\sum\nolimits_n t_{n}^{fd_s}p_{n}.
\label{E_FD2}
\end{eqnarray}
\subsection{Computing model}
We focus on the offloading of tasks to the MEC server that is able to process $F_M$ central processing unit (CPU) cycles per second. Let $c$ is the average number of CPU cycles to process one bit of the task \cite{Fang2022_IoT}. Then, we express the computing delay as:
\begin{equation}
t^{cp}=cD/F_M.
\label{Eq9} 
\end{equation}

\section{Problem formulation}\label{III}
We formulate a resource allocation problem to minimize the sum energy consumption for the offloading of the task from the UE over $N$ (in this paper $N=3$) hops as both the UE and relays are assumed to be energy-constrained while the task processing time meets the maximum required processing time $T_{max}$. This is achieved by optimization of time slots $\boldsymbol{\mathcal{T}}$, transmission power $\boldsymbol{\mathcal{P}}$, and bandwidth allocation $\boldsymbol{\mathcal{B}}$ at individual hops. Hence, the problem is formulated as: 
\begin{eqnarray}
\begin{aligned}
& \boldsymbol{\mathcal{T,P,B}}=
& &\underset{t_n,p_n,b_n}{\text{argmin}} \sum\nolimits_n E_{n}\\
& ~~\text{s.t.}
& & \text{(a)}~ \sum\nolimits_n t_n \leq T_{max}-t^{cp}\\
&
& & \text{(b)}~ t_n>0, \forall n \\
&
& & \text{(c)}~ p_n \leq P_{max}, \forall n \\
&
& & \text{(d)}~ b_n \leq B_{max}, \forall n \\
\end{aligned}
\label{P0}
\end{eqnarray}
where (\ref{P0}a) ensures that task is processed within $T_{max}$, (\ref{P0}b) ensures that each time slot is positive, (\ref{P0}c) limits the transmission power at each hop to $P_{max}$, and (\ref{P0}d) guarantees bandwidth at any does not exceed $B_{max}$.     

\section{Optimization of multi-hop relaying}\label{IV}
In this section, we present the proposed relaying and its optimization. We distinguish three multi-hop relaying cases: \emph{i}) both relays uses HD (labeled as \emph{HD+HD}), \emph{ii}) one relay use \emph{HD} while \emph{FD--Orthogonal} is employed by the other relay (\emph{HD+FD--Orthogonal}), and \emph{iii}) \emph{HD} is exploited by the first relay while \emph{FD--Shared} is used at the second relay (\emph{HD+FD--Shared}). Note that the relaying modes at R1 and R2 can be switched for \emph{ii}) and \emph{iii}) with no impact on derivations presented in the paper. 
We optimize \emph{i})-\emph{iii}) in the following subsections.

\subsection{HD+HD relaying case}\label{IVa}
If both relays employ HD, the task offloading is done during three consecutive time slots (see Fig. \ref{Fig02}). The task is sent first by the UE to the R1 within $t_{1}^{hd}$, then relayed by the R1 to the R2 during $t_{2}^{hd}$, and finally delivered from the R2 to the MEC server in $t_{3}^{hd}$. \textcolor{black}{The benefit of this relaying case is no interference to cope with (such as SI) and relays may support only less complex HD relaying.}

To optimize $\boldsymbol{\mathcal{P}}$ in (\ref{P0}), we express $p_n$ from (\ref{C1_HD}) as a function of $t_n^{hd}$ while assuming $C_n^{hd}=D/t_n^{hd}$ (see (\ref{t_HD})), i.e.:
 \begin{equation}
p_{n}=\frac{K_n}{g_{n}}\left(2^{\frac{D}{t_n^{hd}b_n}}-1\right),
\label{P_HD} 
\end{equation}
where $K_{n}=b_n\left(\sigma+I_b\right)$. Then, the sum energy consumption over all hops is expressed as: 
\begin{multline}
\sum_n E_{n}=\sum_n t_{n}^{hd}p_{n}=\sum_n \frac{t_{n}^{hd}K_n}{g_{n}}\left(2^{\frac{D}{t_{n}^{hd}b_n}}-1\right).
\label{E_HDHD}
\end{multline} 

To optimize $\boldsymbol{\mathcal{B}}$ in (\ref{P0}), since the first derivative of $E_n^{hd}$ with respect to $b_n$ is decreasing with increasing $b_n$, $E_n^{hd}$ at any $n$-th hop is minimized if $b_n=B_{max}$. 
Thus, we can rewrite (\ref{P0}) for optimizing \emph{HD+HD} case as:
\begin{eqnarray}
\begin{aligned}
& \boldsymbol{\mathcal{T}}=
& &\underset{t_{n}^{HD}}{\text{argmin}} \sum\nolimits_n \frac{K_nt_{n}^{hd}}{g_{n}}\left(2^{\frac{D}{t_{n}^{hd}b_n}}-1\right)\\
& ~~\text{s.t.}
& & \text{(\ref{P0}a)}-\text{(\ref{P0}c)}\\
&
& & \text{(d)}~ b_n=B_{max}, \forall n \\
\end{aligned}
\label{P1}
\end{eqnarray}
where (\ref{P1}d) ensures that whole bandwidth is used at all hops. 
\begin{figure}[t!]
\centering
\includegraphics[scale=1]{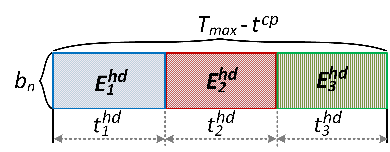}
\caption{Optimization of HD+HD by setting of each $t_{n}^{hd}$.} 
\label{Fig02}
\end{figure}

\begin{lemma}
The optimization problem in (\ref{P1}) and all its constraints are convex with respect to $\boldsymbol{\mathcal{T}}$.
\end{lemma}
\begin{proof}
The Hessian matrix $H$ corresponding to the objective function in (\ref{P1}) is: 

\begin{eqnarray}
\large{H=\begin{bmatrix}L\frac{2^{\frac{D}{t_{1}^{hd}B_{max}}}}{{t_{1}^{hd}}^3g_{1}} & 0 & 0\\ 0 & L\frac{2^{\frac{D}{t_{2}^{hd}B_{max}}}}{{t_{2}^{hd}}^3g_{2}} & 0 \\ 0& 0& L\frac{2^{\frac{D}{t_{3}^{hd}B_{max}}}}{{t_{3}^{hd}}^3g_{3}}\end{bmatrix}}
\end{eqnarray}

\noindent
where $L=(\sigma+I_b)D^2\text{ln}^22/B_{max}$. The entries on the main diagonal of $H$ are positive for $t_{1}^{hd}>0$,  $t_{2}^{hd}>0$, and $t_{3}^{hd}>0$. Since the diagonal matrix $H$ is positive definite, the objective function in (\ref{P1}) is convex.

Further, the constraints (\ref{P0}a), (\ref{P0}b), and (\ref{P1}d) are linear, thus, also convex. Last, using (\ref{C1_HD}) while considering $t_{n}^{hd}=D/C_{n}^{hd}$ (see (\ref{t_HD})), any $p_n$ in (\ref{P0}c) can be rewritten as:

\vspace{-.20cm}
{\small\begin{eqnarray}
t_{n}^{hd}\geq \frac{D}{B_{max}\log_2(1+\frac{g_{n}P_{max}}{K_n})},
\end{eqnarray}}
which is convex (linear) with respect to any $t_{n}^{hd}>0$.
\end{proof}
Since the optimization problem in (\ref{P1}) and all its constraints are convex, any convex optimization method can be used to solve it optimally. We have adopted CVX \cite{CVX}. 

\subsection{HD+FD--Orthogonal relaying case}\label{IVb}
The second case is the combination of \emph{HD} (used by R1) and \emph{FD--Orthogonal} (used by R2), see Fig. \ref{Fig03}. Thus, the task is first sent to R1 during $t_{1}^{hd}$ using $b_1$. Then, the task is simultaneously sent from R1 to R2 and from R2 to MEC server during $t_{2}^{fd_o}=t_{3}^{fd_o}$ (neglecting $\epsilon$ as explained in Section II.B) using $b_2$ and $b_3$, respectively. Note that relaying modes can be switched at R1 and R2. \textcolor{black}{Similarly as in \textit{HD+HD relaying} case, the advantage of \textit{HD+FD--Orthogonal} case is no SI due to relaying, but devices supporting FD relaying have to be employed.}

Like for \emph{HD}, we express $p_n$ as in (\ref{P_HD}) for all hops as the function of time to solve $\boldsymbol{\mathcal{P}}$ in (\ref{P0}). Then, the sum energy consumption is expressed as: 
\begin{multline}
\sum_n E_{n}=E_1^{hd}+E_2^{fd_o}+E_3^{fd_o}= \frac{t_{1}^{hd}K_1}{g_{1}}\left(2^{\frac{D}{t_{1}^{hd}b_1}}-1\right)+\\+\frac{t_{2}^{fd_o}K_2}{g_{2}}\left(2^{\frac{D}{t_{2}^{fd_o}b_2}}-1\right)+\frac{t_{3}^{fd_o}K_3}{g_{3}}\left(2^{\frac{D}{t_{3}^{fd_o}b_3}}-1\right).
\label{E_HDFD1}
\end{multline}

Further, like for \emph{HD}, we can assume $b_1=B_{max}$, since this minimize $E_1^{hd}$. Then, we can reformulate (\ref{P0}) as: 
\begin{eqnarray}
\begin{aligned}
& \boldsymbol{\mathcal{T,B}}=
& &\underset{t_{1}^{hd},t_{2}^{fd_o},b_{2},b_{3}}{\text{argmin}} \left(E_1^{hd}+E_2^{fd_o}+E_3^{fd_o}\right)\\
& ~~\text{s.t.}
& & \text{(a)}~ t_{1}^{hd}+t_{2}^{fd_o} \leq T_{max}-t^{cp}\\
&
& & \text{(b)}~ t_{1}^{hd}>0, t_{2}^{fd_o}>0 \\
&
& & \text{(c)}~ b_{2}+b_{3} \leq B_{max}\\
&
& & \text{(d)}~ b_{2}>0, b_{3}>0\\
&
& & \text{(e)}~ p_n \leq P_{max}, \forall n \\
\end{aligned}
\label{P2}
\end{eqnarray}
where (\ref{P2}a) and (\ref{P2}b) ensure that $T_{max}$ is not violated and the duration of each time slot is positive, respectively, (\ref{P2}c) assures the sum bandwidth at the second and third hops is at most $B_{max}$, (\ref{P2}d) guarantees bandwidth of $b_2$ and $b_3$ is positive, and (\ref{P2}e) is the same as (\ref{P0}c). 

\begin{figure}[t!]
\centering
\includegraphics[scale=1]{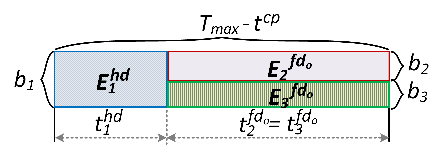}
\caption{Optimization of HD+FD - Orthogonal bandwidth by jointly setting $t_{1}^{hd}$, $t_{2}^{fd_o}$, $b_{2}$, and $b_{3}$.} 
\label{Fig03}
\end{figure}

\begin{lemma}
The optimization problem in (\ref{P2}) and all its constraints are jointly convex with respect to $\boldsymbol{\mathcal{T}}$ and $\boldsymbol{\mathcal{B}}$.
\end{lemma}
\begin{proof}
Similar as in the proof to Lemma 1, the first term  in (\ref{E_HDFD1}) is convex with respect to $t_{1}^{hd}$. To show that the second term in (\ref{E_HDFD1}) is also convex, we prove that for any non-zero $v_1,v_2\in R$ we have:
\begin{gather}
\begin{bmatrix}v_1 &v_2\end{bmatrix}\begin{bmatrix} H_{11} & H_{12} \\H_{21}&H_{22} \end{bmatrix}\begin{bmatrix} v_1\\v_2 \end{bmatrix}>0,
	\label{E_FD1_convex_step1}
\end{gather}

\noindent where $\begin{bmatrix} H_{11} & H_{12} \\H_{21}&H_{22} \end{bmatrix}$ is the Hessian matrix for $E_{2}^{fd_o}$ in (\ref{E_HDFD1}) with respect to included variables $t_{2}^{fd_o}$ and $b_2$. To this end, the left-hand side in (\ref{E_FD1_convex_step1}) is first expanded and rewritten as:
\begin{gather}
\frac{\sigma+I_b}{g_{2}}({b_2}^2v_1^2D^2\text{ln}^2(2)\alpha+{t_{2}^{fd_o}}^2v_2^2D^2\text{ln}^2(2)\alpha+\nonumber\\2v_1v_2\big((t_{2}^{fd_o}b_2)^3(\alpha-1)-t_{2}^{fd_o}b_2\alpha D\text{ln}(2)\times\nonumber\\(t_{2}^{fd_o}b_2-D\text{ln}(2))\big))>0, 	\label{E_FD1_convex_step2}
\end{gather} 
\noindent where $\alpha=2^{\frac{D}{t_{2}^{fd_o}b_2}}$. To prove (\ref{E_FD1_convex_step2}) for any non-zero $v_1,v_2$, according to the Cauchy-Schwarz inequality, it is sufficient to prove that: 
\begin{gather}\label{E_FD1_convex_step3}
	(t_{2}^{fd_o}b_2)^3(2^{\frac{D}{t_{2}^{fd_o}b_2}}-1)-t_{2}^{fd_o}b_22^{\frac{D}{t_{2}^{fd_o}b_2}}\text{ln}(2)\times\nonumber\\(t_{2}^{fd_o}b_2-\text{ln}(2))\geq -t_{2}^{fd_o}b_2\text{ln}^2(2)2^{\frac{D}{t_{2}^{fd_o}b_2}},
\end{gather} 
\noindent or equivalently:
\begin{gather}\label{E_FD1_convex_step4}
	2^{\frac{D}{t_{2}^{fd_o}b_2}}((t_{2}^{fd_o}b_2)^2+2\text{ln}^2(2)-t_{2}^{fd_o}b_2\text{ln}(2))\geq (t_{2}^{fd_o}b_2)^2.
\end{gather} 
\noindent The inequality in (\ref{E_FD1_convex_step4}) always holds since $2^{\frac{D}{t_{2}^{fd_o}b_2}}>1$ and $2\text{ln}^2(2)-t_{2}^{fd_o}b_2\text{ln}(2)>0$. Hence, the Hessian matrix is positive definite. Similarly, the Hessian matrix for $E_{3}^{fd_o}$ in (\ref{E_HDFD1}) is positive definite with respect to the variables  $t_{3}^{fd_o}$ and $b_3$. Last, the constraints (\ref{P2}a)--(\ref{P2}e) are all convex (linear) with respect to the optimization variables.
\end{proof}

Due to convexity of (\ref{P2}) and all its constraints, we can again use CVX as in Section IV.A.

\subsection{HD+FD--Shared relaying case}\label{IVc}
The last case is the one combining \emph{HD} and \emph{FD--Shared} (see Fig. \ref{Fig04}). The offloading follows the same principle as in \emph{HD+FD--Orthogonal}, but the transmissions at the second and third hops overlap also in frequency \textcolor{black}{resulting in a more efficient utilization of communication resources compared to the previous two relaying cases. Still,} the energy consumption at the second hop is affected by SI (see (\ref{C1_FD2})) while the energy consumption at the third hop (at the MEC server) is impacted by interference from R1 (see (\ref{C2_FD2})). Hence, to optimize $\boldsymbol{\mathcal{P}}$ in (\ref{P0}), the transmission power of R1 (i.e., $p_{2}$) is expressed from (\ref{C1_FD2}) while substituting $\{1,2\} \Rightarrow \{2,3\}$ as:
\begin{eqnarray}
p_{2}=\frac{K_2+p_{3}g_{2,2}}{g_{2}}\left(2^{\frac{D}{t_2^{fd_s}b_2}}-1\right).
\label{p1_FD2}
\end{eqnarray} 
Similarly, the transmission power of R2 (i.e., $p_{3}$) is calculated from (\ref{C2_FD2}) assuming $t_2^{fd_s}=t_3^{fd_s}$ and $K_2=K_3$ (since $b_2=b_3$, see Fig. \ref{Fig04}), as:
\begin{eqnarray}
p_{3}=\frac{K_2+p_{2}g_{2,3}}{g_{3}}\left(2^{\frac{D}{t_2^{fd_s}b_2}}-1\right).
\label{p2_FD2}
\end{eqnarray}

\noindent Next, we solve the system of equations formed by (\ref{p1_FD2}) and (\ref{p2_FD2}) in order to express $p_{2}$ and $p_{3}$ independently from each other and only in terms of the other parameters as follows:
\begin{eqnarray}
	p_{2}=\frac{K_2\Gamma}{g_{2}(1-\beta\Gamma^2)}\left(1+\frac{g_{2,2}\Gamma}{g_{3}}\right),\nonumber\\
	p_{3}=\frac{K_2\Gamma}{g_{3}(1-\beta\Gamma^2)}\left(1+\frac{g_{2,3}\Gamma}{g_{2}}\right),
	\label{pr1r2&pm_FD2}
\end{eqnarray}
\noindent where $\beta=\frac{g_{2,3}g_{2,2}}{g_{2}g_{3}}$, $\Gamma=2^{\frac{D}{t_{2}^{fd_s}b_2}}-1$.

Then, the energy consumption of this proposed relaying case is expressed as:
\begin{multline}
\sum_n E_{n}=E_1^{hd}+E_2^{fd_s}+E_3^{fd_s}= \frac{t_{1}^{hd}K_1}{g_{1}}\left(2^{\frac{D}{t_{1}^{hd}b_1}}-1\right)+\\+\frac{t_2^{fd_s}K_2\Gamma}{g_{2}(1-\beta\Gamma^2)}\left(1+\frac{g_{2,2}\Gamma}{g_{3}}\right)+\frac{t_2^{fd_s}K_2\Gamma}{g_{3}(1-\beta\Gamma^2)}\left(1+\frac{g_{2,3}\Gamma}{g_{2}}\right).
\label{E_HDFD2}
\end{multline}

Since the whole $B_{max}$ is used at all hops, as this minimizes the energy consumption (as explained in Section IV.A), the optimization problem in (\ref{P0}) can be formulated as follows:

\begin{eqnarray}
\begin{aligned}
& \boldsymbol{\mathcal{T}}=
& &\underset{t_{1}^{hd},t_{2}^{fd_s}}{\text{argmin}} \left(E_{1}^{hd}+E_{2}^{fd_s}+E_{3}^{fd_s}\right)\\
& ~~\text{s.t.}
& & \text{(a)}~ t_{1}^{hd}+t_{2}^{fd_s} \leq T_{max}-t^{cp}\\
&
& & \text{(b)}~ t_{1}^{hd}>0, t_{2}^{fd_s}>0 \\
&
& & \text{(\ref{P0}c)}, \text{(\ref{P1}d)}  \\
\end{aligned}
\label{P3}
\end{eqnarray}
where the constraints are analogous to those in (\ref{P1}) and (\ref{P2}). 
\begin{figure}[t!]
\centering
\includegraphics[scale=1]{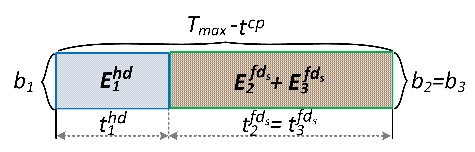}
\caption{Optimization of HD+FD - Shared bandwidth by setting $t_{1}^{hd}$ and $t_{2}^{fd_s}$.} 
\label{Fig04}
\end{figure}

\begin{lemma}
The optimization problem in (\ref{P3}) and all its constraints are convex with respect to $\boldsymbol{\mathcal{T}}$.
\end{lemma}

\begin{proof}
First, $E_1^{hd}$ in (\ref{P3}) is convex with respect to $t_{1}^{hd}$, which is the only variable from $\boldsymbol{\mathcal{T}}$ factoring in $E_{1}^{hd}$. This can be proved similarly as shown in the proof to Lemma 1. 

Next, we show the convexity of $E_2^{fd_s}$ in (\ref{P3}). Since the Hessian matrix of $E_2^{fd_s}$ is too complex, we first decompose $E_2^{fd_s}$ into simpler factors and use the fact that, the multiplication of any positive, strictly decreasing, and convex functions is also convex. This fact can be verified via the equation $(fg)''=f''g+fg''+2f'g'$ for positive, strictly decreasing, and convex arbitrary functions $f$ and $g$. Now by  considering the factors $t_2^{fd_s}K_2\Gamma$ and $\frac{1}{g_{2}(1-\beta\Gamma^2)}$ in $E_2^{fd_s}$ taken from (\ref{E_HDFD2}), both factors are positive, strictly decreasing, and convex with respect to $t_{2}^{fd_s}$. Hence, their multiplication, which yields the term $\frac{t_2^{fd_s}K_2\Gamma}{g_{2}(1-\beta\Gamma^2)}$ in (\ref{P3}), is convex. In addition, the multiplication is also positive and strictly decreasing. 
Next, the term $\Gamma$ and hence $\left(1+\frac{g_{2,2}\Gamma}{g_3}\right)$ in (\ref{E_HDFD2}) is strictly decreasing and convex with respect to $t_2^{fd_s}$. Thus, its multiplication with the  term $\frac{t_2^{fd_s}K_2\Gamma}{g_{2}(1-\beta\Gamma^2)}$, which yields $E_{2}^{fd_s}$ in (\ref{P3}), is convex. 

The third term $E_3^{fd_s}$ in (\ref{P3}) is also convex with respect to $t_2^{fd_s}$, as can be proven analogously to the convexity of $E_2^{fd_s}$. 

Last, the constraints in (\ref{P3}) are also convex (linear) with respect to the optimization variables.
\end{proof}

Since the optimization problem in (\ref{P3}) and all its constraints are convex, we solve it optimally by CVX, analogously as we solve (\ref{P1}) and (\ref{P2}) in Section IV.A and Section IV.B, respectively.  

\begin{figure}[b!]
\centering
\includegraphics[scale=0.85]{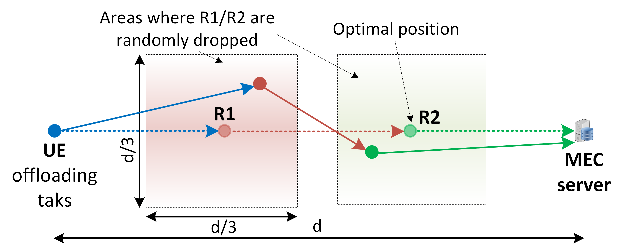}
\caption{Simulation scenario.} 
\label{Fig05}
\end{figure}

\begin{table}[b!]
\footnotesize
\caption{Parameters and settings for simulations}
\label{tab:Tab1}
\centering
\begin{tabular}{|p{1.7cm}|p{1.6cm}|p{1.1cm}|p{2.7cm}| } 
\hline
{\bf Parameter} & {\bf Value} &{\bf Parameter} &{\bf Value} \\
\hline
distance ($d$) & 25-150 m & $I_b$&-150 dBm/Hz \\
\hline
Carrier freq. &2 GHz & $D$&[0.5 2] Mbits \cite{Fang2022_IoT}\\
\hline
$B_{max}$ & 20 MHz & $c$&[1.5 2]x$10^3$ cyc./bit \cite{Fang2022_IoT}\\
\hline
$P_{max}$ & 100 mW &$F_u$ &[0.5 2]x$10^9$ cycles/s \cite{Fang2022_IoT}\\
\hline
$\sigma$& -174 dBm/Hz& $F_M$&40x$10^9$ cycles/s \cite{Fang2022_IoT}\\
\hline
\end{tabular}
\label{tab}
\end{table}

\begin{figure*}[t!]
\centering
\begin{subfigure}[b]{0.3\textwidth}
\includegraphics[width=\textwidth]{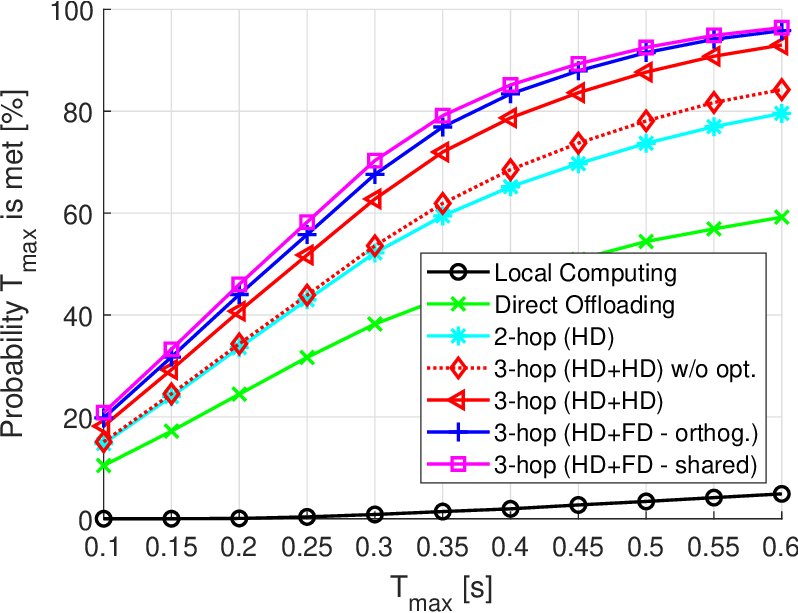}
\caption{}
\label{}
\end{subfigure}
~ 
\begin{subfigure}[b]{0.3\textwidth}
\includegraphics[width=\textwidth]{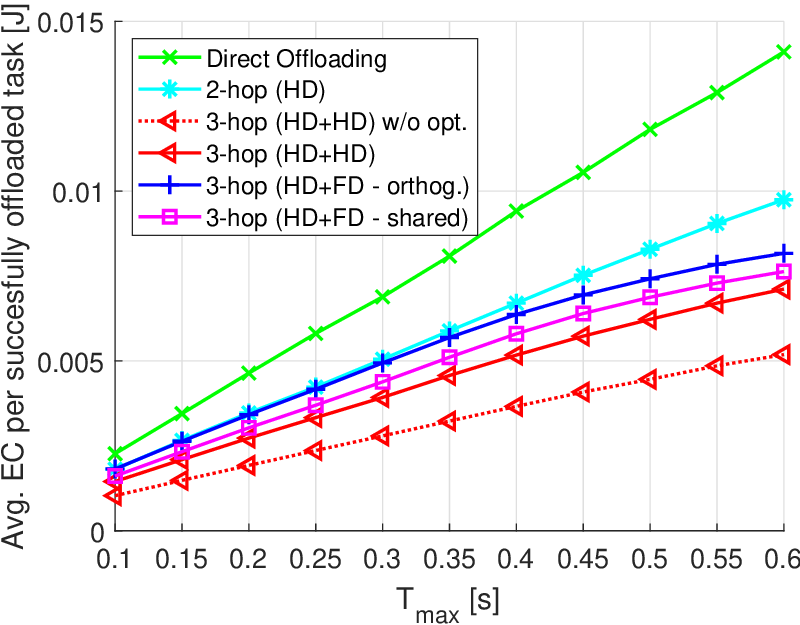}
\caption{}
\label{}
\end{subfigure}
~ 
\begin{subfigure}[b]{0.3\textwidth}
\includegraphics[width=\textwidth]{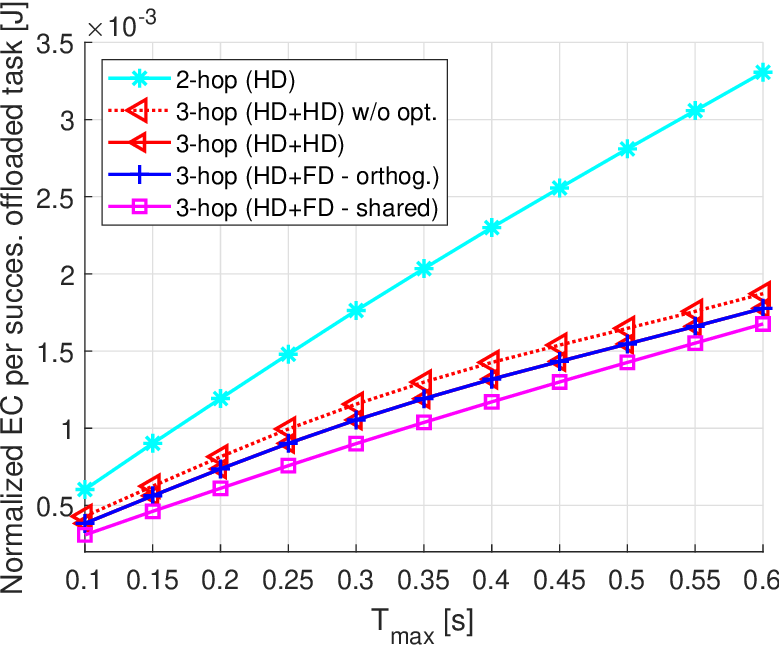}
\caption{}
\label{}
\end{subfigure}
\caption{Effect of $T_{max}$ on: a) probability that $T_{max}$ is met, b) average energy consumption per successfully offloaded task, c) normalized energy consumption when only tasks successfully offloaded within $T_{max}$  by direct offloading are considered.}
\label{Fig06}
\end{figure*}

\section{Performance evaluation}\label{VIII}
In this section, we first describe the simulation models and parameters and then analyze the performance of individual relaying schemes.

\subsection{Simulation models and parameters}
For the performance evaluation, we assume the scenario with one UE offloading tasks via two relays to the MEC server (see Fig. \ref{Fig05}). To get statistically valid results, we average out the results over 200 000 drops. Within each drop, we randomly generate: \emph{i}) tasks parameters $D$ and $c$, \emph{ii}) distance between the UE and the MEC server between 25 and 150 m, and \emph{iii}) positions of the relays within the areas shown in Fig. \ref{Fig05}. This random generation of the relays' positions substitutes the relay selection process and each drop represents a case with relay at random positions. We adopt general modified COST 231 Hata path loss model at 2 GHz. All important simulation parameters are summarized in Table I. 

We compare the results of multi-hop offloading with: \emph{i}) local computing at the UE with computing power ($F_u$) randomly generated in each drop, \emph{ii}) direct offloading to the MEC server without relays, \emph{iii}) offloading via only relay working in HD usually considered in the related state-of-the-art works (denoted as 2-hop (HD)), \emph{iv}) offloading via multi-hop (i.e., two relays) with HD at both relays while relaying is not optimized, as considered in \cite{2023_TVT}-\cite{2023_TNSM} (denoted as 3-hop (HD+HD) w/o opt.)

\subsection{Results}
In Fig. \ref{Fig06}a, we investigate the probability that the tasks are successfully processed within $T_{max}$. Following the intuition, the probability that the tasks are successfully processed increases with $T_{max}$ for all investigated schemes, as there is more available time to process the tasks. 
Since we target tasks with relatively high requirements on computation, the local computing is not efficient with only up to 5\% probability that tasks are processed within $T_{max}$. The direct offloading of tasks to MEC server significantly increases the probability that $T_{max}$ is met (up to 60\%). The introduction of relaying notably improves the performance of the offloading so that 2-hop relaying and 3-hop relaying without optimization leads to the probability of successful task processing within $T_{max}$ up to 79.5\% and 84.2\%, respectively, if $T_{max}=0.6$ s.

Now, let's discuss the probability of successful processing within $T_{max}$ of the proposed optimized multi-hop relaying cases described in Section \ref{IV}. The superior performance is provided by the multi-hop relaying combining \emph{HD} and \emph{FD--shared}, outperforming the direct offloading, 2-hop relaying, and conventional 3-hop relaying without optimization in terms of probability of $T_{max}$ being met by up to 99.7\%, 41.3\%, and 38\%, respectively. Among the proposed multi-relaying approaches, the worst performance is observed for \emph{HD+HD}, as it is less spectrum efficient. Still even this scheme outperforms direct computing, 2-hop relaying, and conventional 3-hop relaying without optimization in terms of probability $T_{max}$ is met by up to 74.3\%, 23.3\%, and 20.4\%, respectively.

In Fig. \ref{Fig06}b, we analyze the average energy consumed per the task successfully processed within $T_{max}$. We demonstrate that the average energy consumption increases with $T_{max}$ since, generally, offloading can take longer, thus consuming more energy. As expected, the highest energy consumption is spent by the direct offloading. The energy consumption is decreased by more than 30\% if 2-hop relaying is introduced. Further significant decrease in the energy consumption (nearly 3 times compared to the direct offloading) is observed for the multi-hop relaying. If the multi-hop relaying is not optimized, it can even consume less energy than the proposed multi-hop relaying cases. This is due to the fact that the optimization of multi-hop relaying allows to accommodate also more demanding tasks within $T_{max}$ that cannot be processed successfully without the proposed optimization. However, these more demanding tasks cost more energy during offloading. As a result, the average energy consumption per successfully offloaded task is increased.

To make the comparison of energy consumed per task fair, in Fig. \ref{Fig06}c, we show the normalized energy consumption, i.e., \textcolor{black}{the energy consumption over only those tasks successfully offloaded by all compared schemes.} Note that the direct offloading is omitted in Fig. \ref{Fig06}c to keep a reasonable scale of the y-axis and also since the energy consumption for the direct offloading is, in fact, the same as in Fig. \ref{Fig06}b. Fig. \ref{Fig06}c demonstrates that the best performance is yielded by proposed multi-hop relaying combining \emph{HD} with \emph{FD--shared} as it decreases the energy consumption when compared to 2-hop and 3-hop relayings without optimization by up to 51.2\% and 28\%, respectively.

\section{Conclusions}\label{VIII}
In this paper, we have focused on offloading of highly computationally demanding tasks to MEC via multi-hop relaying. We have introduced and optimized several relaying cases combining half and full duplex relaying at individual relays. We have demonstrated multi-hop relaying improves the offloading experience while decreases the energy consumption of energy-constrained devices involved in the relaying. This paper is an initial work and a joint optimization of multi-hop relaying and relay selection should be carried out in the future.

\section*{Acknowledgement}
This work was supported by the Czech Science Foundation under the grant no. 23-05646S.

\end{document}